\documentclass[letterpaper, 10 pt, conference]{ieeeconf}                               
\IEEEoverridecommandlockouts                              

\expandafter\ifx\csname pdfoptionalwaysusepdfpagebox\endcsname\relax\else
\fi

\usepackage{microtype}
\usepackage{parskip}
\usepackage{ragged2e}
\usepackage[dvipsnames]{xcolor}
\usepackage[pdfencoding=auto, psdextra]{hyperref}
\hypersetup{
colorlinks=true,
linkcolor=Blue,
filecolor=Blue,      
urlcolor=Black,
citecolor = Blue,
}
\usepackage{soul}

\usepackage{listings}

\definecolor{codegreen}{rgb}{0,0.6,0}
\definecolor{codegray}{rgb}{0.5,0.5,0.5}
\definecolor{codepurple}{rgb}{0.58,0,0.82}
\definecolor{backcolour}{rgb}{1,1,1}

\lstdefinestyle{mystyle}{
    backgroundcolor=\color{backcolour},   
    commentstyle=\color{codegreen},
    keywordstyle=\color{magenta},
    numberstyle=\tiny\color{codegray},
    stringstyle=\color{codepurple},
    basicstyle=\ttfamily\footnotesize,
    breakatwhitespace=false,         
    breaklines=true,                 
    captionpos=b,                    
    keepspaces=false,                 
    numbers=left,                    
    showspaces=false,                
    showstringspaces=false,
    showtabs=false,                  
    tabsize=1,
}
\let\labelindent\relax 
\usepackage[shortlabels]{enumitem}

\usepackage{mdframed}

\usepackage{graphicx}
\usepackage{float}
\usepackage[]{subcaption}

\newcommand{\lwdth}{1.5pt}
\mdfsetup {
        linewidth = \lwdth,
        innerleftmargin = 0,
        innerrightmargin = 0,
        innertopmargin = 0,
        innerbottommargin = 0,
        }

\usepackage{svg}

\usepackage{multirow}
\usepackage{rotating}
\usepackage{makecell}
\usepackage{booktabs}
\makeatletter
\def\thickhline{%
  \noalign{\ifnum0=`}\fi\hrule \@height \thickarrayrulewidth \futurelet
   \reserved@a\@xthickhline}
\def\@xthickhline{\ifx\reserved@a\thickhline
               \vskip\doublerulesep
               \vskip-\thickarrayrulewidth
             \fi
      \ifnum0=`{\fi}}
\makeatother

\newlength{\thickarrayrulewidth}
\usepackage{makecell} 
\usepackage{tablefootnote}
\usepackage{array}
\newcommand{\PreserveBackslash}[1]{\let\temp=\\#1\let\\=\temp}
\newcolumntype{C}[1]{>{\PreserveBackslash\centering}p{#1}}
\newcolumntype{R}[1]{>{\PreserveBackslash\raggedleft}p{#1}}
\newcolumntype{L}[1]{>{\PreserveBackslash\raggedright}p{#1}}
\usepackage{supertabular}
\usepackage{tabularx}
\definecolor{darkgray}{rgb}{0.663,0.663,0.663}

\usepackage{algorithm}
\usepackage{algorithmicx}
\usepackage[italicComments=true,commentColor=darkgray]{algpseudocodex}

\usepackage{amsthm}
\usepackage{amsmath}
\usepackage[capitalize]{cleveref} 
\crefformat{equation}{(#2#1#3)} 
\usepackage{amssymb,amsfonts, mathtools}
\usepackage{bbm}
\usepackage{bm}
\usepackage{cuted}
\usepackage{etoolbox}
\AtBeginEnvironment{gather}{\setcounter{equation}{0}}
\makeatletter
\def\@endtheorem{\endtrivlist}
\makeatother

\makeatletter
\newtheoremstyle{mytheorem}
  {3pt}
  {3pt}
  {\itshape}
  {}
  {\itshape\bfseries}
  {:}
  {.5em}
  {\thmname{#1}\normalfont\itshape\thmnumber{\@ifnotempty{#1}{ }#2}%
   \thmnote{ {\the\thm@notefont(\itshape #3)}}}
\makeatother
\theoremstyle{mytheorem}

\AtBeginEnvironment{gather}{\setcounter{equation}{0}}
\newtheorem{lemma}{Lemma}

\newtheorem{remark}{Remark}
\newtheorem{proposition}{Proposition}

\newtheorem{theorem}{Theorem}

\makeatletter
\newtheoremstyle{def}
  {3pt}
  {3pt}
  {}
  {}
  {\itshape\bfseries}
  {:}
  {.5em}
  {\thmname{#1}\normalfont\itshape\thmnumber{\@ifnotempty{#1}{ }#2}%
   \thmnote{ {\the\thm@notefont(\itshape #3)}}}
\makeatother
\theoremstyle{def}
\newtheorem{definition}{Definition}

\makeatletter
\newtheoremstyle{assumps}
  {}
  {}
  {\itshape}
  {  }
  {\bfseries}
  {:}
  {\newline}
  {\thmname{#1}{\@ifnotempty{#1}{}}%
   \thmnote{{\the\thm@notefont(\itshape #3)}}}
\makeatother
\theoremstyle{assumps}
\newtheorem{assumptions}{Assumptions}

\makeatletter
\newtheoremstyle{consts}
  {}
  {}
  {\itshape}
  {}
  {\bfseries}
  {:}
  {\newline}
  {\thmname{#1}{\@ifnotempty{#1}{}}%
   \thmnote{ {\the\thm@notefont(\itshape #3)}}}
\makeatother
\theoremstyle{consts}
\newtheorem{constraints}{Safety Constraints}

\newcommand{\mf}[1]{\mathfrak{#1}}
\newcommand{\mc}[1]{\mathcal{#1}}
\newcommand{\mck}{\mc{K}}
\newcommand{\mbr}{\mathbb{R}}
\newcommand{\ith}{i^{\text{th}}}
\newcommand{\kth}{k^{\text{th}}}

\newcommand{\mrmx}{\mathrm{x}}

\newcommand{\hsighat}{\Hat{h}_\Sigma}
\newcommand{\phihat}{\Hat{\phi}(\ell,\beta)}
\newcommand{\hsig}{{h}_\Sigma}

\newcommand{\minhstar}{\min_{\star\in\mathcal{I}}{h_\star}}

\newcommand{\hdot}[1]{\frac{\partial h_{#1}(\mrmx)}{\partial \mrmx}}
\newcommand{\hdiffdot}{\sum_{\star\in\mathcal{I}\setminus \{1\}} {\hdot{\star}} - \hdot{1}}
\newcommand{\hsumdot}{\sum_{\star\in\mathcal{I}} {\hdot{\star}}}

\newcommand{\hdiff}{\sum_{\star\in\mathcal{I}\setminus \{1\}}{h_{\star}(\mrmx)} - h_1(\mrmx)}
\newcommand{\phihatbetadot}{\frac{\partial \phihat}{\partial \mrmx }}

\newcommand{\mru}{\mathrm{u}}
\newcommand{\mrustar}{\mathrm{u}_\star}

\newcommand{\mrm}[1]{\mathrm{#1}}
\newcommand{\umax}{u_{\text{max}}}

\newcommand{\liabnorm}{L^1[a,b]}
\usepackage{upgreek}
\usepackage{wasysym}
\usepackage{dsfont}

\newcommand{\E}[1]{\mathbb{E}\left[#1\right]}

\usepackage{pgf, tikz}
\usetikzlibrary{decorations.fractals,shadows,shapes.symbols,spy}
\usepackage{ifthen}
\usetikzlibrary{shapes,arrows,external,calc, automata, arrows.meta,patterns, tikzmark,decorations.pathreplacing,calligraphy, positioning} 
\usepackage{pgfplots}
\tikzstyle{ipe stylesheet} = [
  ipe import,
  even odd rule,
  line join=round,
  line cap=butt,
  ipe pen normal/.style={line width=0.4},
  ipe pen heavier/.style={line width=0.8},
  ipe pen fat/.style={line width=1.2},
  ipe pen ultrafat/.style={line width=2},
  ipe pen normal,
  ipe mark normal/.style={ipe mark scale=3},
  ipe mark large/.style={ipe mark scale=5},
  ipe mark small/.style={ipe mark scale=2},
  ipe mark tiny/.style={ipe mark scale=1.1},
  ipe mark normal,
  /pgf/arrow keys/.cd,
  ipe arrow normal/.style={scale=7},
  ipe arrow large/.style={scale=10},
  ipe arrow small/.style={scale=5},
  ipe arrow tiny/.style={scale=3},
  ipe arrow normal,
  /tikz/.cd,
  ipe arrows, 
  <->/.tip = ipe normal,
  ipe dash normal/.style={dash pattern=},
  ipe dash dotted/.style={dash pattern=on 1bp off 3bp},
  ipe dash dashed/.style={dash pattern=on 4bp off 4bp},
  ipe dash dash dotted/.style={dash pattern=on 4bp off 2bp on 1bp off 2bp},
  ipe dash dash dot dotted/.style={dash pattern=on 4bp off 2bp on 1bp off 2bp on 1bp off 2bp},
  ipe dash normal,
  ipe node/.append style={font=\normalsize},
  ipe stretch normal/.style={ipe node stretch=1},
  ipe stretch normal,
  ipe opacity 10/.style={opacity=0.1},
  ipe opacity 30/.style={opacity=0.3},
  ipe opacity 50/.style={opacity=0.5},
  ipe opacity 75/.style={opacity=0.75},
  ipe opacity opaque/.style={opacity=1},
  ipe opacity opaque,
]
\definecolor{red}{rgb}{1,0,0}
\definecolor{blue}{rgb}{0,0,1}
\definecolor{green}{rgb}{0,1,0}
\definecolor{yellow}{rgb}{1,1,0}
\definecolor{orange}{rgb}{1,0.647,0}
\definecolor{gold}{rgb}{1,0.843,0}
\definecolor{purple}{rgb}{0.627,0.125,0.941}
\definecolor{gray}{rgb}{0.745,0.745,0.745}
\definecolor{brown}{rgb}{0.647,0.165,0.165}
\definecolor{navy}{rgb}{0,0,0.502}
\definecolor{pink}{rgb}{1,0.753,0.796}
\definecolor{seagreen}{rgb}{0.18,0.545,0.341}
\definecolor{turquoise}{rgb}{0.251,0.878,0.816}
\definecolor{violet}{rgb}{0.933,0.51,0.933}
\definecolor{darkblue}{rgb}{0,0,0.545}
\definecolor{darkcyan}{rgb}{0,0.545,0.545}
\definecolor{darkgreen}{rgb}{0,0.392,0}
\definecolor{darkmagenta}{rgb}{0.545,0,0.545}
\definecolor{darkorange}{rgb}{1,0.549,0}
\definecolor{darkred}{rgb}{0.545,0,0}
\definecolor{lightblue}{rgb}{0.678,0.847,0.902}
\definecolor{lightcyan}{rgb}{0.878,1,1}
\definecolor{lightgray}{rgb}{0.827,0.827,0.827}
\definecolor{lightgreen}{rgb}{0.565,0.933,0.565}
\definecolor{lightyellow}{rgb}{1,1,0.878}
\definecolor{black}{rgb}{0,0,0}
\definecolor{white}{rgb}{1,1,1}
\pgfplotsset{width=10cm,compat=1.9}
\tikzstyle{block} = [draw, fill=green!20, text centered, rectangle,
minimum height=0.8cm,minimum width=6em, align=right]

\usepackage[mode=buildnew,subpreambles=true]{standalone}

\definecolor{mplblue}{RGB}{31, 119, 180}
\definecolor{mplred}{RGB}{214, 39, 40}
\definecolor{mplgrey}{RGB}{127, 127, 127}

\usepackage{nicematrix}

\crefformat{assumption}{Assumption #1}

\crefformat{equation}{(#2#1#3)} 

\crefformat{assumption}{Assumption #1}

\makeatletter
\def\@opargbegintheorem#1#2#3{\trivlist
   \item[]{\itshape #1\ #2\ (#3)}\\}
\makeatother

\usepackage[norule,hang]{footmisc}
\usepackage{nccmath}
\newcommand{\scs}[1]{{\scshape{#1}}}
\newcommand{\hlrv}[1]{{\color{black}{#1}}}

\newcommand{\Prob}[2][]{%
  \operatorname{Pr}\ifthenelse{\isempty{#1}}{}{\left(#2\right)}%
}

\newcommand{\mrmdrift}{{\upmu}_h}
\newcommand{\mrmdiff}{{\upsigma}_h}
\newcommand{\tg}{{\operatorname{T}}_{\mathrm{g}}}
\newcommand{\tmax}{{\operatorname{T}}_{\mathrm{max}}}
\newcommand{\tcont}{{\operatorname{T}}_{\mathrm{u}}}

\usepackage{nicefrac}
\newcommand{\Bsighat}{\Hat{B}_\Sigma}
\newcommand{\Bmrmdrift}{{\upmu}_B}

 \SetSymbolFont{largesymbols}{bold}{OMX}{txex}{b}{n}

 \usepackage{silence}
\usepackage{bold-extra} 
\usepackage{ragged2e}
\justifying

\usepackage{eso-pic}
\AddToShipoutPicture*{\small \sffamily\raisebox{1.2cm}{\hspace{1.9cm}Copyright \copyright \ 2024 IEEE}}

\begin{document}

\title{Safe Collective Control under Noisy Inputs and Competing Constraints via Non-Smooth Barrier Functions\textsuperscript{*} \thanks{\textsuperscript{*}This work received support in part from the Office of Naval Research (Grant No. N000141712622) and from a seed grant awarded by Northrop Grumman Corporation.}
\thanks{\textsuperscript{1}C. Enwerem and J.S. Baras are with the Department of Electrical \& Computer Engineering and the Institute for Systems Research, University of Maryland, College Park, MD 20742, USA.
Emails: \{\texttt{enwerem, baras}\}@umd.edu.}
}
\author{{Clinton Enwerem\textsuperscript{1}}
\and
{John S. Baras\textsuperscript{1}}}

\maketitle 

\begin{abstract}
     We consider the problem of safely coordinating ensembles of identical autonomous agents to conduct complex missions with conflicting safety requirements and under noisy control inputs. Using non-smooth control barrier functions (CBFs) and stochastic model-predictive control as springboards, and by adopting an extrinsic approach where the ensemble is treated as a unified dynamic entity, we devise a method to synthesize safety-aware control inputs for uncertain collectives. Drawing upon stochastic CBF theory and recent developments in Boolean CBF composition, our method proceeds by smoothing a Boolean-composed CBF and solving a stochastic optimization problem where each agent's forcing term is restricted to the affine subspace of control inputs certified by the combined CBF. For the smoothing step, we employ a polynomial approximation scheme, providing evidence for its advantage in generating more conservative yet sufficiently-filtered control inputs than the smoother but more aggressive equivalents produced from an approximation technique based on the log-sum-exp function. To further demonstrate the utility of the proposed method, we present an upper bound for the expected CBF approximation error, along with results from simulations of a single-integrator collective under velocity perturbations. Lastly, we compare these results with those obtained using a naive state-feedback controller lacking safety filters.
\end{abstract}


\IEEEpeerreviewmaketitle 

\section{Introduction}
Intelligent cyber-physical systems --- such as teams of aerial robots or autonomous vehicle platoons --- face a fundamental challenge: they must safely navigate complex environments while efficiently completing their objectives. This challenge is non-trivial because the sub-tasks that constitute the ensemble's objective may often be varied and conflicting, making it imperative to conduct appropriate trade-offs between tasks. The inherent uncertainty in each agent's model and operational vicinity also imposes additional constraints on the multiagent system, making safe multiagent control an acute problem.
\begin{figure}[t!]\centering{\includegraphics[interpolate, width=.65\columnwidth]{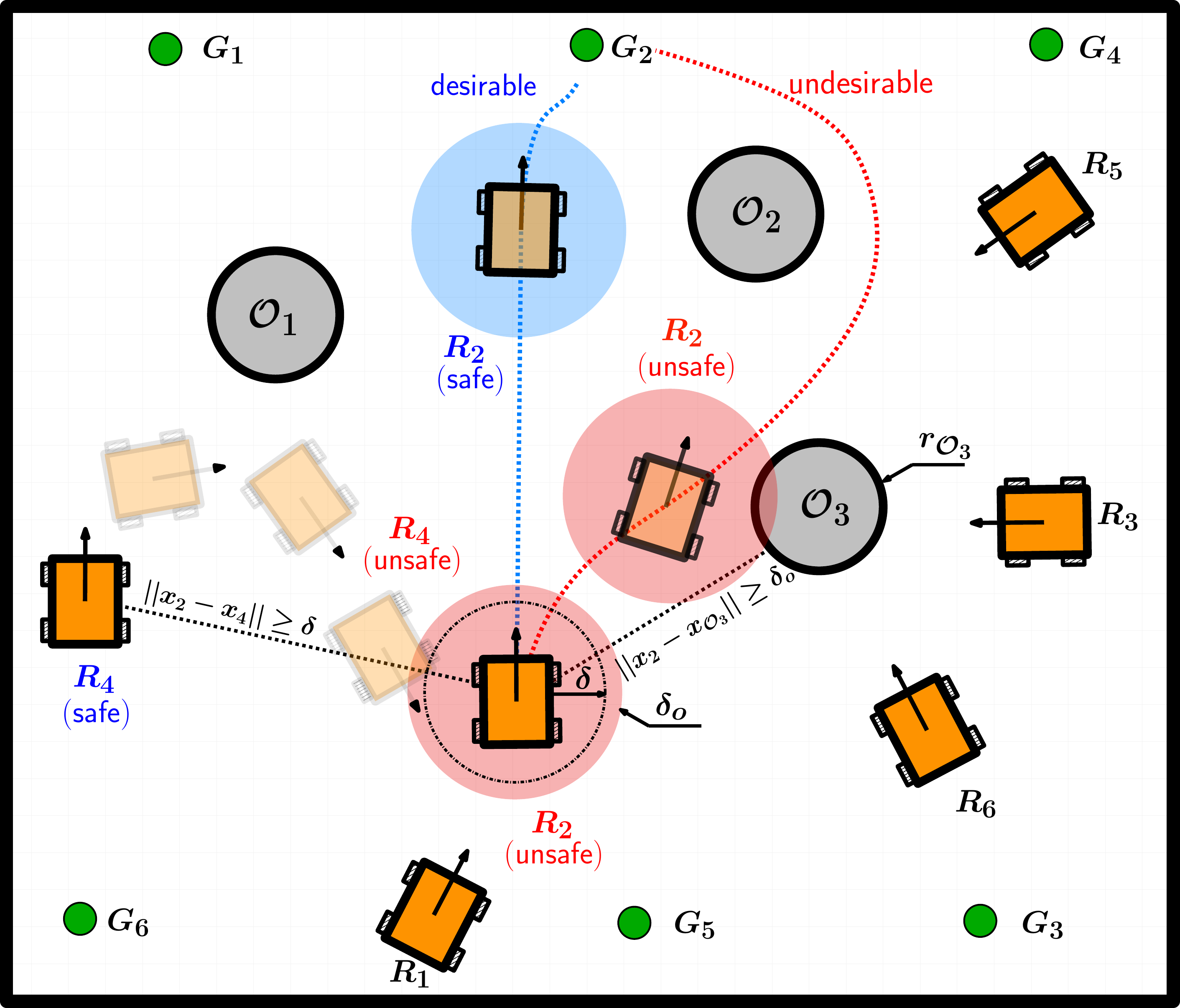}}
    \caption{\textbf{\hlrv{Motivating example}}: A multi-agent reach-avoid mission with complex safety requirements. Here, the dynamic agents (depicted as mobile robots in orange ($R_i$)) must navigate to specified goal positions (green circles ($G_i$)) while avoiding multiple obstacles (gray circles ($\mathcal{O}_i$)) and inter-agent collisions. Figure annotations describe associated notions of safety, with agent arrows indicating direction of travel.} 
    \label{fig:reachavoidmas}
\end{figure}
Control barrier functions (CBFs) --- an arguably de facto technique for safety-critical controller synthesis --- have excelled in many applications in autonomous driving, robotics \cite{glotfelterNonsmoothBarrierFunctions2017}, and allied fields. However, conventionally, most of the associated results in CBF-enabled control are founded on deterministic assumptions, making it challenging to apply them in uncertain settings. Moreover, designing CBFs for complex missions with competing safety requirements typically entails a composition of distinct barrier functions, an operation that may render the resulting function non-smooth, so that traditional results of local Lipschitz continuity fail to apply.%

In composing barrier functions, one also runs the risk of potentially losing safe controller existence guarantees. \hlrv{Again, this is due to the fact that many CBF-derived control synthesis algorithms generate filtered control signals by solving quadratic programs (QPs) --- or some optimization problem with state and control costs --- via gradient descent or other derivative-based techniques. As such, even if one could come up with CBFs that sufficiently capture all the given safety requirements at hand, there might not exist any control sequence for which the constraints of the resulting CBF quadratic program are satisfied \cite{molnarComposingControlBarrier2023}}. \hlrv{Moreover, while an existence of non-negative Lagrange multipliers for each CBF constraint has been shown to guarantee feasibility of the resulting QP \cite{molnarComposingControlBarrier2023}, such conditions are for the deterministic case and do not apply in the stochastic setting}.

\hlrv{Regarding CBF composition, a number of studies have presented techniques for capturing diverse safety requirements in the deterministic setting (see \cite{glotfelterNonsmoothBarrierFunctions2017,molnarComposingControlBarrier2023})}. These methods typically involve computing the pointwise maximum (resp. minimum) of all barrier functions that encode the safety requirements, equivalent to an OR (resp. AND) Boolean logical operation. Because the foregoing operations yield non-smooth functions, however, synthesizing safe controllers becomes a challenge, since certifying control inputs via barrier functions requires computing the latter's Lie derivatives along the control and drift vector fields. Despite these challenges, recent work  \cite{molnarComposingControlBarrier2023} has demonstrated that safe control inputs can still be synthesized by approximating the non-smooth CBF (NCBF) via the log-sum-exp (LSE) function. They further posit that the resulting CBF is bounded tightly by asymptotically-decaying terms containing the NCBF, under such a smoothing scheme. In the stochastic case and for multiagent systems, however, no such analysis or guarantees exist, to the best of our knowledge. With the recent emergence of results extending CBF theory to stochastic systems \cite{so2023almost}, we now have the tools to address the foregoing research gaps.

\subsection{Scientific Contributions}
Leveraging new-found ideas on Boolean composition of CBFs with smooth approximation, we report results on synthesizing safe controllers for ensembles under complex tasking and noisy inputs. In contrast to similar work \cite{molnarComposingControlBarrier2023}, however, we apply a polynomial smoothing function in the smoothing step and synthesize control inputs via a \textit{stochastic} quadratic program. Finally, we present two numerical examples to demonstrate the proposed technique and validate the approach through comparisons with a naive state-feedback controller and an LSE-approximated CBF-enabled control scheme. 

\subsection{Outline}
Notations that appear throughout the paper have been organized on \cref{tab:notnom}. Occasionally, several notations appearing there will contain a subscript, a superscript, or both, but they will be elucidated if unclear from the context. In the sections that follow, we briefly discuss the mathematical elements upon which our work is founded, viz. non-smooth CBFs (\cref{sec:cbfrev}), Boolean composition (\cref{sec:compncbf}), and stochastic model-predictive control (SMPC) (\cref{sec:rmpc}). These sections set the background for our theoretical results (\cref{sec:theores}), with numerical simulations (\cref{sec:numex}), associated findings (\cref{sec:resdis}), the article's conclusion (\cref{sec:conc}), and the appendices appearing last. 
\begin{table}[t]
    \centering
    \setlength\arraycolsep{0pt}
    \caption{\sc Notation \& Nomenclature}
    \label{tab:notnom}
    \begin{tabularx}{\columnwidth}{L{2.1cm}p{6cm}}
    \toprule
      Symbol & {\centering Description}\\
      \midrule
      \midrule
      $x_i \in \mathfrak{X} \subset \mathbb{R}^d$ & Node- or agent-level state variable\\
      $u_i \in \mathfrak{U} \subset \mathbb{R}^m$ & Node-level forcing term\\
      ${\zeta}_d$ & Vector $\begin{bsmallmatrix} \zeta, & \zeta, & \dots & \zeta \end{bsmallmatrix}^{\top} \in \mathbb{R}^d$\\
      $\mathbb{I}_d$ & $d\times d$ identity matrix\\
      $\lvert\lvert {x_i} \rvert\rvert_p$ & $p$-norm (set to two where $p$ is omitted)\\
      $[{A_{ij}}]$; ${\operatorname{tr}(A)}$; $A^{\top}$ & Matrix ($A$) with elements $A_{ij}$; trace; transpose\\
      $||H||_{\liabnorm}$ & $\int_{a}^{b}{|H(\tau)|}d\tau$, $H$ \hlrv{Lipschitz} continuous on $[a,b]$\\
      ${\Lambda}([\star_i]_{i\in\mathcal{I}})$ & Block diagonal matrix with $\ith$ block $\star_i$\\
      $\mathbb{R}_+$; $\mathbb{Z}_{\ge 0}$, $\mathbb{Z}_{+}$ & Set of positive reals; non-negative \& positive integers\\
      $C^k$ & Set of $k$-times continuously-differentiable functions\\
      $f(\cdot)$; $g(\cdot)$ & Drift and control vector fields\\
      $\otimes$; $\alpha(\cdot)$; & Kronecker product; Extended $\kappa$-class function\\
      $h(\cdot)$ & Control barrier function (CBF)\\
      $L_fh(\cdot)$; $L_gh(\cdot)$; & Lie derivatives of $h$ along $f$ and $g$\\
      $\mathbb{E}$; $\operatorname{Pr}(E_i)$ & Expectation operator; Probability of event $E_i$\\
    \bottomrule
    \end{tabularx}
\end{table}

\section{Non-Smooth Control Barrier Functions}
\label{sec:cbfrev}
To preface our discussion on NCBFs, we consider a collective comprising $N$ identical agents, each evolving according to nonlinear control-affine dynamics with noisy inputs given by the following equation (with subscript $i$ denoting the $\ith$ agent): 
\begin{align}
 \label{eq:contaffss}
        \dot{x}_i = f(x_i) + g(x_i)u_i + K_w(x_i)w_i,
 \end{align}
 with admissible control signals taking values in the set, $\mathfrak{U}$, of piecewise continuous and absolutely-integrable functions. $w_i\in\mbr^d$ represents the realization of an unknown disturbance (with a known probability distribution) that is scaled by a state-dependent and positive-definite gain matrix, $K_w\in\mbr^{d\times d}$. \hlrv{The equation in \cref{eq:contaffss} coincides with a stochastic differential equation comprising a \textit{diffusion} term equal to $K_w(x_i)$} and an affine-in-control \textit{drift} term. \hlrv{A few useful definitions follow}.

\begin{definition}[Safe Set]
    \hlrv{Let $x_i$ satisfy \cref{eq:contaffss}}. The safe set, $\mf{C}$, with non-empty interior, $\operatorname{Int}(\mf{C})$, and boundary, $\partial\mf{C}$, is the zero super-level set of the sufficiently smooth map, $h(\cdot): \mf{C} \rightarrow \mbr,$ that is, $\mf{C} = \{x_i \in \hlrv{\mathfrak{X}} : h(x_i) \ge 0\}$.
\end{definition}

\begin{definition}[Control Barrier Function \cite{ames_control_2019}]
\label{def:cbf}%
\noindent 
Assume agent-level dynamics of the form \cref{eq:contaffss}, with $f$ and $g$ locally Lipschitz, and suppose there exists a function, $\alpha$, belonging to the family of locally Lipschitz continuous extended $\mck$-class functions\footnote{An extended $\mck$-class function is a strictly increasing function that vanishes at the origin.} on $\mbr^d$. The $C^2$ function, $h: \mathfrak{C}\subset \hlrv{\mathfrak{X}} \rightarrow \mbr_+\cup\{0\}$, is a control barrier function\footnote{To be precise, $h$ as delineated in \cref{def:cbf} refers to the special class of CBFs known as zeroing or \textit{zero control} CBFs \cite{ames_control_2017}, so named for their property of vanishing at the boundary of the safe set, i.e., $\partial\mathfrak{C}$. \hlrv{They are distinct from \textit{reciprocal} CBFs (RCBFs) that instead grow unbounded as the system state approaches $\partial\mathfrak{C}$}.} for the system in \cref{eq:contaffss}  if there is a $u_i$ for which the following property holds for all $x_i \in \hlrv{\mathfrak{X}}$ satisfying $h(x_i) > 0$, with $\gamma > 0$:
\begin{equation}
\label{eq:cbfppty}\hlrv{
 L_fh + L_ghu_i + \frac{1}{2}{\operatorname{tr}\left(K_w^{\top}\frac{\partial^2h}{\partial x_i^2}K_w\right)} \ge -\alpha(\gamma, {h(x_i)})}.
\end{equation}
\end{definition}
\normalsize
The family of controllers satisfying the foregoing condition are important because they establish the forward invariance of $\mf{C}$, for $\mf{C}$ compact for all $t \ge 0$. \hlrv{In particular, under certain conditions on the drift and diffusion terms corresponding to $dh/dt$ (the exact derivative of $h$), then $\hlrv{\operatorname{Pr}}(x_i(t) \in \mf{C} \ \forall \ t \geq 0) = 1$, provided that $x_i(0) \in \operatorname{Int}(\mathfrak{C})$ (see the recent paper \cite{so2023almost}, \cref{thm:soalmostsureconv}, and \cref{prop:finvcsigma} with its proof for the details)}. This result thus enables the synthesis of safe controllers via stochastic quadratic programming, set here as the expectation minimization problem
\begin{align}
\label{eq:cbfqp}
&\quad \min_{u_{i} \in \mathfrak{U}} J = \E{\sum_{\tau = 0}^{T-1}{c\left(x_i(\tau), u_i(\tau)\right)}}\\
& \text{s.t.} \ \hlrv{\operatorname{Pr}}\bigg(L_fh + L_ghu_i + \frac{1}{2}{\operatorname{tr}\left(K_w\frac{\partial^2h}{\partial x_i^2}K_w\right)} \ge -\alpha(\gamma, {h})\bigg)\nonumber\\
&\quad \ge 1-\delta_h\nonumber,
\end{align}
for $\delta_h \in (0, 1)$ small. The symbol $T \in \mathbb{Z}_+$ represents some time horizon, and $c(x_i(t), u_i(t))$ is a positive-definite cost function (typically made up of quadratic terms in $u_i$ and $x_i$) chosen to capture a desired weighting of the control inputs corresponding to each component of $x_i$. With this primary setting in place, we transition to discussing Boolean composition of CBFs.

\section{Composing CBFs via Boolean Logic}
\label{sec:compncbf}
Composing barrier certificates for complex safety requirements via Boolean logical operations appears in several studies in the literature \cite{molnarComposingControlBarrier2023,glotfelterNonsmoothBarrierFunctions2017,ahmadi_safe_2019}. Besides Boolean logic, however, it is theoretically possible to enforce global safety requirements through distinct CBFs by passing them as independent constraints to a QP. Nevertheless, non-smooth compositions --- comprising some combination of maximization, minimization, or negation --- are standard, for reasons already alluded to. These operations are respectively analogous to the Boolean logical operators of OR $(\vee)$, AND $(\wedge)$, and NOT $(\neg)$, and the set-theoretic operations of union $(\bigcup)$, intersection $(\bigcap)$, and complement $({}^c)$. More concretely, suppose we have a finite number of independent safety constraints completely indexed by the set, $\mathcal{I}$. We are interested in computing a common CBF with attendant control inputs that drive each agent in the ensemble only to configurations that respect all $|\mathcal{I}|$ safety constraints. Denote by \hlrv{$\mrmx = [[x_i^{\top}]_{i=1}^N]^{\top} \in \mathbb{R}^{Nd}$ and $\mathrm{u} = [[u_i^{\top}]_{i=1}^N]^{\top} \in \mathbb{R}^{Nm}$} the ensemble analogs of the state and control vectors. The ensemble-level safe set, hereafter denoted as $\mathfrak{C}_\Sigma$, with corresponding (zeroing) barrier function, $h_\Sigma$, is the set 
\begingroup
\begin{align}
\label{eq:cstar} 
\mathfrak{C}_\Sigma &= \{\mrmx \in \mathbb{R}^{Nd}: h_\Sigma(\mrmx) \ge 0\}\\
\label{eq:cstarb} 
\quad &= \bigcap_{|\mathcal{I}|} \mf{C}_\star, \ \text{for} \ \star = 1, 2, \dots, |\mathcal{I}|\\
\quad &= \left\{\mrmx \in \mathbb{R}^{Nd}: \min_{\star \in \mathcal{I}}{h_\star(\mrmx)} \ge 0\right\},  
\end{align}
\endgroup
where $\hsig$ is a conjunction of all independent CBFs, over $\mathcal{I}$, given by the expression:
\begin{equation}
    \label{eq:compcbfs}
    h_\Sigma(\mrmx) = \bigwedge_{\star=1}^{|\mathcal{I}|}h_\star(\mrmx).
\end{equation}
Similarly, one can construct the OR analog of \cref{eq:cstarb} (resp. \cref{eq:compcbfs}) through a union (resp. disjunction) over $\mathcal{I}$. Even more complicated safety requirements can be captured via successive Boolean logical compositions \cite{glotfelterNonsmoothBarrierFunctions2017,molnarComposingControlBarrier2023}, as we illustrate in \cref{sec:masf}. 

\section{Approximating NCBFs via Polynomial Smoothing}
Since $h_\Sigma$ is non-smooth and, thus, unfit for quadratic optimization, in this section, we discuss a parametric polynomial smoothing technique for finding an approximation of $h_\Sigma$, hereafter denoted as $\hsighat$. Rewriting $h_\Sigma = \minhstar$ (with $\mrmx$ suppressed) as
\begin{equation}
    \label{eq:hsigalt}
    h_\Sigma(\mrmx) = -\frac{1}{2}{\big(\ell\phi(\ell) - \ell^\prime\big)},
\end{equation}%
where the expression on the right-hand side of \cref{eq:hsigalt} follows from a known equivalent expression for the min. function \cite{sahiner2018smoothing}, with $\ell =  \hdiff$ and $\ell^\prime = \sum_{\star\in\mc{I}}h_\star(\mrmx)$. The function $\phi: \mathbb{R} \rightarrow \{-1, 1\}$ characterizes the discontinuity of $h_\Sigma$ and is given by%
\begin{align}
    \label{eq:phi}
    \phi(\ell) = \begin{cases}
        1, \ & \ell \ge 0\\
        -1, \ & \ell < 0.
    \end{cases}
\end{align}
Smoothing $h_\Sigma$ thus reduces to the task of redefining \cref{eq:phi} to address the jump discontinuity at $\ell = 0$. In particular, if we redefine $\phi$ (hereafter denoted as $\Hat{\phi}(\ell, \beta)$ to make its distinction from \cref{eq:phi} clear) as
\begingroup
\small
\begin{align}
        \label{eq:phihat}
        \Hat{\phi}(\ell, \beta) = \begin{cases}
        1, \ &\ell > \beta\\
        M_k(\ell, \beta), \ &-\beta \le \ell \le \beta\\
        -1, \ &\ell < -\beta,
    \end{cases}
\end{align}
\endgroup
where $\beta$ is the smoothing parameter and $M_k(\ell, \beta) = a_0(\beta) + a_1(\beta)\ell + a_2(\beta)\ell^2 + \dots + a_{p}(\beta)\ell^{p}$, we arrive at the following expression for $\Hat{h}_\Sigma(\mrmx)$: 
\begin{equation}
    \label{eq:hsigaltfin}
    \Hat{h}_\Sigma(\mrmx) = -\frac{1}{2}{\big(\ell\Hat{\phi}(\ell, \beta) - \ell^\prime\big)}.
\end{equation}
The choice of $p$ in $M_k$ depends on the desired level of continuous differentiability for $\hsighat$, i.e., $p=2k+1, k \in \mathbb{Z}_{+}$, for $\hsighat \in C^k$. Denoting the ensemble drift and control vector fields respectively as $F$ and $G$, we can write the corresponding ensemble Lie derivatives (omitting higher derivatives for brevity) as
\hlrv{
$$\begin{array}{ll}
    \label{eq:liehsigf}
    L_F\Hat{h}_\Sigma(\mrmx) = \frac{\partial\hsighat}{\partial \mrmx}F(\mrmx) \ \text{and} \
    L_G\Hat{h}_\Sigma(\mrmx) = \frac{\partial\hsighat}{\partial \mrmx}G(\mrmx),
\end{array}$$
where $F(\mrmx)$ and $G(\mrmx)$ are given respectively by
$$
\begin{array}{ll}
&\begin{bmatrix}f^{\top}({x}_1), & f^{\top}({x}_2), & \dots, & f^{\top}({x}_{N})\end{bmatrix}^{\top} \in \mathbb{R}^{{N}d} \ \text{and} \ \\ 
&\bigwedge\big(\begin{bmatrix}g({x}_1),& g({x}_2),& \dots,& g({x}_{N})\end{bmatrix}\big) \in \mathbb{R}^{{N}d \times {N}m},
\end{array}
$$} 
with
\begingroup
\small
\begin{align}
    \label{eq:hsighatdot}
    \frac{\partial \hsighat}{\partial \mrmx} &= -\frac{1}{2}\Bigg[\bigg(\hdiffdot\bigg)\phihat\\
    &+ \Bigg(\hdiff\Bigg)\phihatbetadot - \hsumdot \Bigg].\nonumber
\end{align}
\endgroup
\normalsize

\section{SMPC with Polynomial-Approximated NCBFs}
\label{sec:rmpc}
\normalsize
\enlargethispage{\baselineskip}
We segue now to generating control commands for each agent via SMPC. In particular, having constructed a valid and smooth CBF from the given independent safety constraints, we can now synthesize control inputs by solving the (ensemble-level) stochastic quadratic program
\begingroup
    \begin{subequations}
    \label{eq:cbfqprob}
   \begin{align}
    &{\min_{\mathrm{u} \in \mathbb{R}^{Nm}}} \mathbb{E}\bigg[{\sum_{\tau = 0}^{\tcont-1}||\mathrm{u}(\tau) - \mathrm{u}_{d}(\tau)||^2}\cdot \big(1-\mathds{1}_{\{||\mrmx - \mrmx_g|| \le \varepsilon_g\}}\big)\bigg]\nonumber\\
    \label{eq:chancconst}
    &\text{s.t.} \quad \hlrv{\operatorname{Pr}}\big(L_F\Hat{h}_\Sigma(\mrmx) + L_G\Hat{h}_\Sigma(\mrmx)\mathrm{u}(\tau) \nonumber\\
    & \quad + \frac{1}{2}{\hlrv{\operatorname{tr}}\big((K_w\otimes I_N)^{\top}\frac{\partial^2\hsighat(\mrmx)}{\partial \mrmx^2}(K_w\otimes I_N)\big)} \nonumber\\
    & \quad + \alpha(\gamma, {\hsighat(\mrmx)}) \ge 0\big)\ge1-\delta_h\\
    \label{eq:ubounds}
    & \quad -\umax\cdot1_{Nm} \le \mrm{u}(\tau) \le \umax\cdot1_{Nm},  
    \end{align}
    \end{subequations}
\endgroup
with $\mathrm{u}_{d} \in \mathbb{R}^{Nm}$ denoting the nominal and unsafe ensemble-level control input, and where $\tcont$ is the MPC control horizon. We include an indicator function, $\mathds{1}_{\{||\mrmx - \mrmx_g|| \le \varepsilon_g\}}$, to prevent the accumulation of costs after all agents have reached positions within a goal region defined by the closed ball centered at $\mrmx_g$ with radius, $\varepsilon_g$. The agents' respective goal locations are given by each element of the goal vector, $\mrmx_g = [x_{i,g}^{\top}]^{\top}$. Finally, we also place bounds on the control inputs to respect physical actuation limits \cref{eq:ubounds}, characterized by a positive scalar, $\umax$. Due to the hard input constraints and chance constraints, SMPC problems of the form \cref{eq:cbfqprob} are generally intractable \cite{mesbahStochasticModelPredictive2016}. \hlrv{Thus}, a few assumptions (presented next) are in order, to guarantee the problem's feasibility.
\begin{assumptions}[]
\begin{enumerate}[start=1,label={(A\arabic*)},leftmargin=*]
\setlength{\itemsep}{-3pt}
    \item[]
    \item \label{asm:fullstate}
    The dynamical system for the $\ith$ agent is fully observable, so that access to full state information is assumed.
    \item The disturbance vector is zero-mean Gaussian with bounded covariance, i.e., $w_i\sim \mathcal{N}(0_d, [\Sigma_{w_{ij}}])$, with $0 < \Sigma_{w_{ii}} \le \sigma_{\mathrm{max}} \ \forall \ i \in [1, 2, \dots, d]$, where $\sigma_{\mathrm{max}} \in \mbr_+$ represents a bound on the noise. This assumption is important to argue for the convexity of the set of safe controls using the log-concavity of the multivariate normal distribution.
    \item $\mathrm{u}_{d}$ is parameterized by the disturbance as $\mathrm{u}_d(t) = \mru_{f}(t) + (K_w\otimes I_N)\mrm{w}$, with $\mrm{w} = [[w_i^{\top}]_{i=1}^N]^{\top} \in \mathbb{R}^{Nd}$ denoting the ensemble disturbance vector and where $\mru_f$ is a simple feedback law (see \cref{sec:masf} for its form).
\end{enumerate}
\end{assumptions}
These assumptions together with a restriction of the control inputs to an invariant set \cref{eq:chancconst} and the quadratic cost guarantee the feasibility of \cref{eq:cbfqprob}  \cite{mesbahStochasticModelPredictive2016}, since they enable a transformation of the chance constraints into convex constraints via a \hlrv{logarithm} operation. \hlrv{One can then solve the resulting QP using optimization software with quadratic programming support.}  

\section{Theoretical Results}
\label{sec:theores}
\hlrv{\subsection{Upper Bound on the Expected CBF Approximation Error}}\label{ssec:cbferrbound}
To quantify the accuracy of the CBF approximation with respect to the smoothing parameter, we introduce the following lemma, which we will invoke in the proposition to follow.
\begin{lemma}
    \label{lem:phihat}
    Let the functions $\phihat$ and $\phi$ be as defined in \cref{eq:phihat} and \cref{eq:phi}, respectively. Then, for $\beta > 0$ and on ${\mathcal{I}_{\beta}} = [-\beta, \beta]$, the following inequality holds:
    \begin{equation}
        \label{eq:phihatminusphi}
        \E{|| \phi(\ell) - \phihat||_{\liabnorm}} = \frac{15}{4}\beta.
    \end{equation}
\end{lemma}%
\begin{proof}
    Since $\phihat$ and $\phi(\ell)$ coincide for $\ell < -\beta$ and $\ell > \beta$, when $\beta > 0$, and the $\liabnorm$ norm is equivalent to the integral, $\int_{-\beta}^{\beta}{| \phi(\ell) -  \phihat|} d\ell$, on ${\mathcal{I}_{\beta}}$, it is sufficient to prove that \cref{eq:phihatminusphi} holds on ${\mathcal{I}_{\beta}}$. By straightforward integration on ${\mathcal{I}_{\beta}}$, with $k = 2$ in \cref{eq:phihat} (see the corresponding coefficients for $M_2(\ell, \beta)$ on \cref{tab:params}), and using the definition of the $\liabnorm$ norm and the fact that $\E{c} = c$, if $c$ is a constant, we arrive at the expression in \cref{eq:phihatminusphi}.
\end{proof}
\begin{proposition}
    Let $\hsighat$ be defined as in \cref{eq:hsigaltfin}. Define the CBF approximation error as $\hat{e} = \hsighat - \minhstar$. Then, under the polynomial smoothing scheme, with $k = 2$ in \cref{eq:phihat}, we have the following \hlrv{inequality} for the expected CBF approximation error, $\E{||\hat{e}||_{\liabnorm}}$:
    \begin{equation}
        \label{eq:upbound}
        \E{||\hat{e}||_{\liabnorm}} \le \frac{15\beta^2}{8}.
    \end{equation}
\end{proposition}
\begin{proof}
    See Appendix A.
\end{proof}
\begin{remark}
Since the upper bound on the CBF approximation error consists of terms quadratic in $\beta$, successively higher values of $\beta$ may lead to smoother approximations under the polynomial smoothing scheme at the cost of losing information about the original function. Conversely, the approximations for $\beta$ values on $(0, \frac{2\sqrt{2}}{\sqrt{15}})$ may be comparatively less smooth while being less prone to distortion. It turns out, in fact, that no perceptible change in the smoothing performance is recorded for $\beta$ values on $(0, 1]$ (see \cref{fig:betaeffect}), with safe controls assured nonetheless (see \cref{fig:usafe}). 
\end{remark}
\hlrv{\subsection{Forward Invariance of the Ensemble-Level Safe Set}}\label{ssec:hsigvalid}
\hlrv{Considering that our method relies on the synthesis of control inputs via a polynomial-smoothed NCBF (i.e., $\hsighat$), that itself defines the ensemble-level safe set ($\mf{C}_\Sigma$), we need to verify that $\hsighat$ is indeed a valid CBF. The following result from \cite{so2023almost} specifies conditions that guarantee almost-sure safety via stochastic \textit{zeroing} CBFs that will be useful in establishing the validity of $\hsighat$ in our next result (Proposition 2).}%
\noindent%
\hlrv{
\begin{theorem}[Conditions for Almost-Sure Safety via Zeroing CBFs (Corollary 11 in \cite{so2023almost})]
    \label{thm:soalmostsureconv}%
    Assume there exists a function $h: \mathfrak{X} \rightarrow \mathbb{R}$ and a $\mck$-class function, $\alpha_\star$, where, for all $x_i \in \mathfrak{X}$ satisfying $h(x_i)>0$, there exists a control $u_i \in \mathfrak{U}$ such that\footnote{Note that $\alpha_\star$ ($\mck$-class function) and $\alpha$ (extended $\mck$-class function), while similar, denote separate objects throughout this article.}
    \begin{equation}
    \label{eq:soalmostsureconv}
    {\tilde{\mu}}_i-\frac{{\tilde{\sigma}_i}^2}{h(x_i)} \geq-h^2(x_i) \alpha_\star(h),
    \end{equation}
    with ${\tilde{\mu}}_i$ and ${\tilde{\sigma}}_i$, respectively denoting the drift and diffusion terms corresponding to $dh/dt$, with $x_i$ evolving according to \cref{eq:contaffss}. Then, for all $t \geq 0, \operatorname{Pr}\left(x(t) \in \mathfrak{C}\right)=1$, provided that $x(0) \in \operatorname{Int}(\mathfrak{C})$.
\end{theorem}
\begin{proposition}[Forward Invariance of $\mf{C}_\Sigma$]\label{prop:finvcsigma}
    Provided that $\mrmx(0)\in\mathfrak{C}_\Sigma$, if there exists a $\mrm{u}:[0, \infty)\rightarrow\mf{U}$ satisfying the ensemble analog of \cref{eq:soalmostsureconv} (given in Appendix B), then $\mrmx(t)\in\mathfrak{C}_\Sigma\ \forall \ t \geq 0$, with probability 1. 
\end{proposition}
}
\begin{proof}
    This proof follows from an invocation of \cref{thm:soalmostsureconv}, from $\mf{C}_\Sigma$'s definition, and from the smoothness of $\hsighat$. \hlrv{We provide the details in Appendix B}.
\end{proof}
\hlrv{
\begin{remark}
    Although the proof for Proposition 2 establishes the invariance of $\mf{C}$ for all $t \ge 0$, we note here that we require safety to be assured only for a finite interval, since the controls are computed (and applied) for $t \in \{0, 1, \ldots, \tmax-1\}$ (see the limits of the summation operator in \cref{eq:cbfqp,eq:cbfqprob}).
\end{remark}
}
\begin{algorithm}[htb]
\caption{\hlrv{Safety-Aware Collective Control via Polynomial-Smoothed NCBFs and Stochastic Quadratic Programming}}\label{alg:ncbfps}
\begin{algorithmic}[1]%
\Require $N$, $\tcont$, $\mrmx_g$, $\varepsilon_g$, Dynamics: $f, g, K_w, w_i$.%
\State Max. mission time (in number of time steps): $\tmax$.
\State $\tg \gets 0$.
\Repeat
    \State $\mrmx_0$ = \scs{StateEstimator}(). 
    \State $\mrmx_c \gets \mrmx_0$ \Comment{Store current state}.
    \State Sample ensemble disturbance realization, $\mrm{w}.$
    \For{$\tau \gets 1$ to $\tcont$}
             \State Compute $\mru_d (\tau) = ({\mrmx_c} - \mrmx_{g}) + (K_w\otimes I_N)\mrm{w}$.
             \State $\mrmx_c \gets F(\mrmx_c)+G(\mrmx_c)\mru_d(\tau)$.
    \EndFor
    \State Solve stochastic QP \cref{eq:cbfqprob} for $\{\mrustar(t)\}_{t=0}^{\tcont-1}$.
    \State Send $\mrustar(0)$ to agents' actuators.
    \For{$\tau \gets 1$ to $\tcont-2$}
        \State $\mru(\tau-1) \gets \mrustar(\tau)$.
    \EndFor
    \State $\mru(\tcont-1)\gets$ $\mru_{\text{init}}$ \Comment{Initial element of $\mru$}. %
    \If{$\tg = \tmax$}
        \State \textbf{break} \Comment{Exit loop}.
    \EndIf
    \State $\tg \gets \tg + 1$.
\Until{$||\mrmx_c - \mrmx_g|| \le \varepsilon_g$.}
\end{algorithmic}
\end{algorithm}

\section{Multi-Agent Safety Notions}
\label{sec:masf}
While different notions for safety exist, each depending on the specific task requirements at hand, we define safety in this article using the two constraints provided next (see \cref{fig:reachavoidmas} for a visual summary of the tasks under consideration).
\begin{constraints}[]
\begin{enumerate}[start=1,label={(C\arabic*)},leftmargin=*]
\setlength{\itemsep}{-2pt}
    \item[] 
    \item an inter-agent collision avoidance constraint, which requires each agent to maintain a distance of $\delta$ meters (m) from other agents, and
    \item an obstacle avoidance constraint that ensures agents do not collide with $N_o$ static circular obstacles of radius $r_{o}$ in the mission field by keeping a distance of $\delta_o > \delta$ m, with $\delta_o$ at least twice the width of each agent\footnote{We will assume that these distances are detected by range sensors fitted on each agent. Obstacles of disparate sizes can also be considered with an appropriate modification of $\delta_o$.}. 
\end{enumerate}
\end{constraints}
The aforementioned conditions can be enforced separately by the CBFs ($i$ signifies the agent index) 
\begin{align}
    \label{eq:indcbfs}
    h_i^\delta &= ||x_i - x_j||^2 - \delta^2\\
    h_i^o &= ||x_i - x_{k}^o||^2 - \delta_o^2,  
\end{align}
where ${x_{k}^o}$ is the absolute position of the $\kth$ obstacle. Using \cref{eq:compcbfs}, we select the following CBF candidate for the ensemble $h_\Sigma = h_\delta \wedge h_o$, with $h_\delta$ and $h_o$ given, respectively, as
\begingroup
    \label{eq:hdelta}
    \begin{align}
    \label{eq:deltah}
    &\bigwedge_{i=1}^N\bigwedge_{j=1+1}^Nh_i^\delta \ \text{and}\ 
    \bigwedge_{i=1}^N\bigwedge_{k=1}^{N_o} h_i^o.
    \end{align}
\endgroup
\normalsize
 The entire safe set $\mf{C}_\Sigma$ is, thus, $\mf{C}_\Sigma = \bigcap_{\star\in\{\delta, o\}}\mf{C}_\star,$ with corresponding CBF, $\min_{\star\in\{\delta, o\}}{h_\star}$. Finally, we will assume that $\mru_f$ in (A3) evolves according to the following feedback control law:  $\mru_f(t) = {\mrmx}(t) - \mrmx_{g}$. This nominal controller is thus adjusted by the safety filter provided by the CBF only as needed, to adhere to the specified safety requirements.

\section{Numerical Examples}
\label{sec:numex}
We discuss two mission scenarios in this section. In the first mission, the agents must each navigate to their respective goal locations within a planar mission space that includes a single obstacle whose position is known; in the second, multiple obstacles are introduced. 
For both examples, we consider single-integrator systems with dynamics of the form\begingroup\begin{subequations}\label{eq:sintd}
    \begin{eqnarray}
            \dot{p}_{x_i} &=& (v_{x_i}+w_{x_i})\\
            \dot{p}_{y_i} &=& (v_{y_i}+w_{y_i}), i=1, 2, \dots, N,
    \end{eqnarray}
    \end{subequations}\endgroup
with control inputs perturbed by $w_{x_i}$ and $w_{y_i}$, the $x$ and $y$ components of the normally-distributed disturbance, $w_i$. $p_{x_i}$ and $p_{y_i}$ are the $\ith$ agent's position in the plane along the abscissa and ordinate axes, respectively, with corresponding linear velocities, $v_{x_i}$ and $v_{y_i}$. It is trivial to see that \cref{eq:sintd} is equivalent to the control-affine form in \cref{eq:contaffss}, with $f = 0_2$ and $g = \mathbb{I}_2$, and where $ u_i = [v_{x_i}, v_{y_i}]^{\top}$ and $x_i = [p_{x_i}, p_{y_i}]^{\top}$ are the $\ith$ agent's control and state vectors, respectively. For convenience, we will take the disturbance terms to encode the nonlinearity in the agents' models. Furthermore, the agents are tasked with the multi-objective problem set forth in \cref{sec:numex}. \cref{tab:params} lists values of the parameters we used during simulations. Assuming identical agents, we can then write the ensemble-level dynamics for the single-integrator collective as $\dot{\mrmx} = \mathrm{u} + (K_w\otimes I_N)\mathrm{w}$, with ${\mrmx} = [p_{x_1}, p_{y_1}, p_{x_2}, p_{y_2}, \dots, p_{x_N}, p_{y_N}]^{\top}$ and $\mathrm{u} = [v_{x_1}, v_{y_1}, v_{x_2}, v_{y_2}, \dots, v_{x_N}, v_{y_N}]^{\top}$.
\begin{table}[htb]
    \centering
    \setlength\arraycolsep{2pt}
    \caption{\sc Simulation Parameters}
    \label{tab:params}
    \begin{NiceTabular}{C{1cm}C{0.8cm}C{1.5cm}C{1cm}}
    \toprule
    Parameter & Value & Parameter & Value \\
    \midrule
    \midrule
    $\beta$ & $[1, \frac{1}{2}, \frac{1}{10}]$ & $a_3$ & $-\frac{5}{4\beta^3}$ \\
    $[\delta, \delta_0]$ & $[0.14, 0.15]$ & $a_5$ & $\frac{3}{8\beta^5}$ \\
    $\delta_h$ & $0.01$ & $K_w$ & $\mathbb{I}_2$ \\
    $a_1$ & $\frac{15}{8\beta}$ & $\alpha(\gamma, h)$ & $\gamma^3 h$ \\
    $a_{2k, k\in\mathbb{Z}_{\ge 0}}$ & $0$ & $\Sigma_w$ & $0.1\cdot\mathbb{I}_d$ \\
    \bottomrule
    \end{NiceTabular}
\end{table}
In \cref{alg:ncbfps}, we provide pseudocode for implementing the proposed safety-aware control scheme. $\tg$ is the number of time steps it takes for all agents to reach positions within their respective goal sets. \hlrv{Its value is unknown; however, for finite goal reachability, we require that $\tg$ be bounded by a finite number, $\tmax \in \mathbb{Z}_+$. The \scs{StateEstimator} sub-routine returns an estimate of the current ensemble state.}
\section{Results \& Discussions}
\label{sec:resdis}
\begin{figure*}[htb]
    \centering
    \def\probFigW{.9}
    \hspace{-80pt}
    \begin{subfigure}[b]{\probFigW\columnwidth}
        \centering
    \includegraphics[trim= 0pt 0pt 0pt  0pt, clip,width=.75\columnwidth]{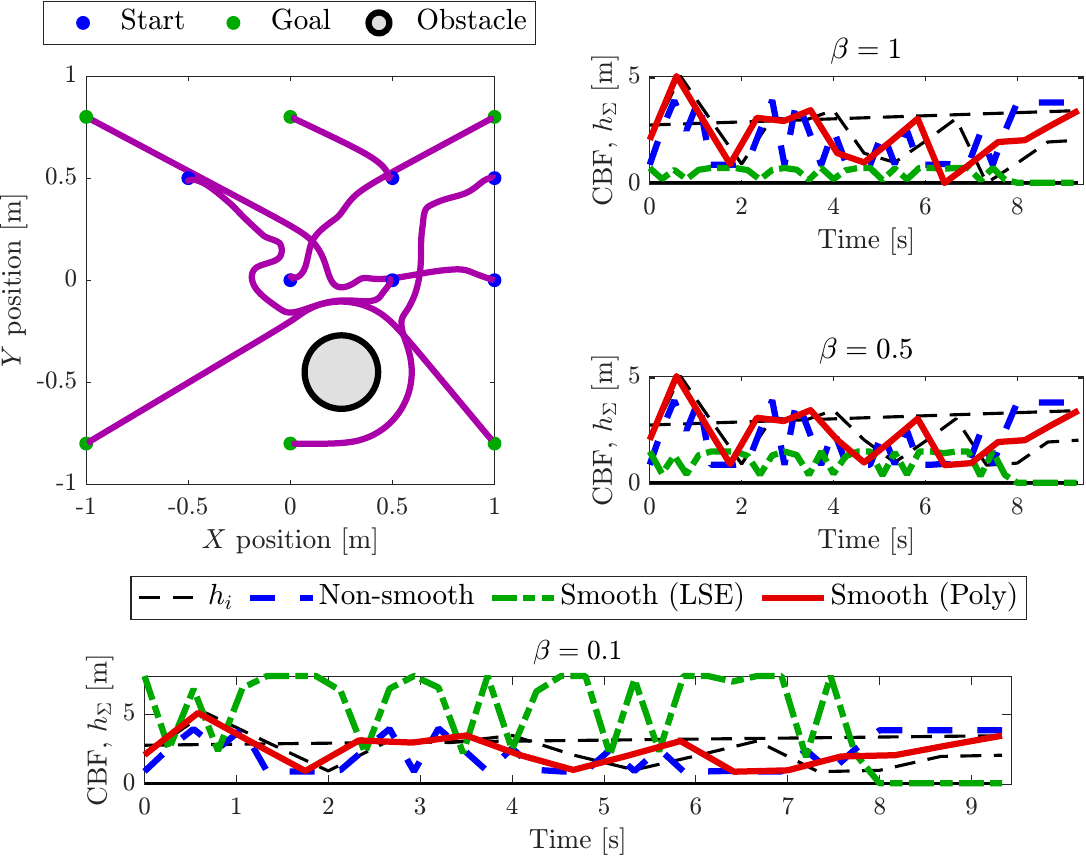}
        \caption{}
        \label{fig:betaeffectsingobs}
    \end{subfigure}
    \hspace{-50pt}
    \begin{subfigure}[b]{\columnwidth}
    \centering
    \includegraphics[width=.75\columnwidth]{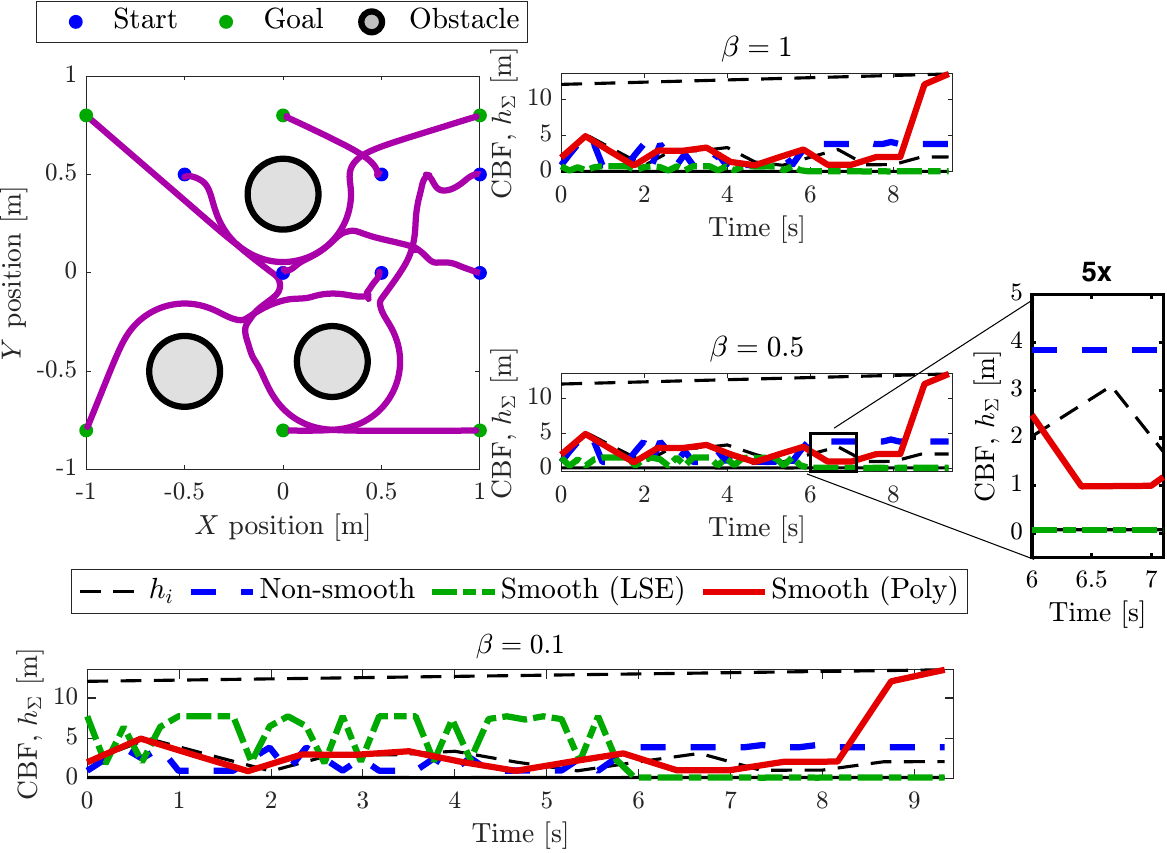}
    \caption{}
    \label{fig:betaeffectmultobs}
    \end{subfigure}
    \hspace{-35pt}
    \begin{subfigure}[b]{.6\columnwidth}
    \centering
    \includegraphics[trim= -13pt 0pt 0pt  0pt, clip,width=.9\columnwidth, scale=10]{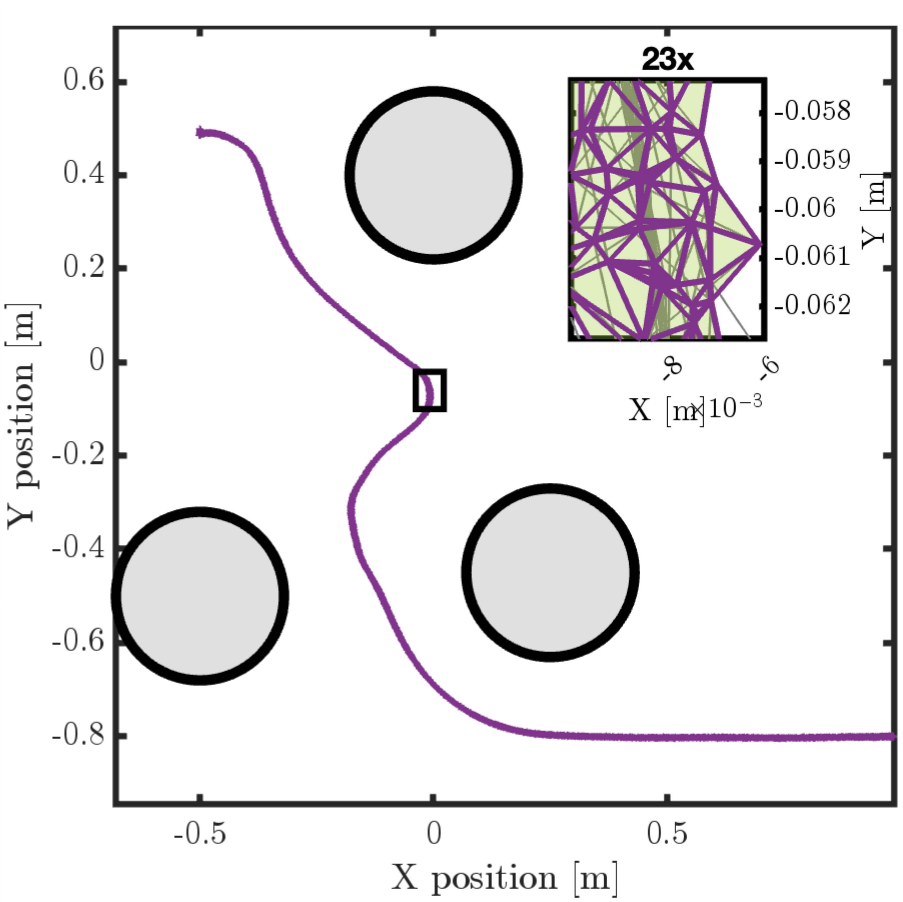}
    \caption{}
    \label{fig:trajenv}
    \end{subfigure}
    \hspace{-60pt}
    \caption{\textbf{Comparing controller reach-avoid performance}: (\subref{fig:betaeffectsingobs}) Time-evolution of agent trajectories and CBFs showing the effect of smoothing technique and smoothing parameter ($\beta$) on $h_\Sigma$, with a horizontal line at $h_\Sigma = 0$ for reference. (\subref{fig:betaeffectmultobs}) \hlrv{Results for the multiple-obstacle example. From the inset on the right, we notice that the LSE-smoothed CBF assumes a zero value rather abruptly on the interval $6\le t\le7$.} (\subref{fig:trajenv}) Envelope of 100 safe trajectories (with magnified inset) for agent 3.}
    \label{fig:betaeffect}
\end{figure*}

\subsection{Safety-Critical Control under Varied Safety Specifications}
In \cref{fig:betaeffectsingobs}, we plot the collective's trajectories and associated CBFs, with both smoothing schemes and for the single-obstacle mission. Similar plots corresponding to the multiple-obstacle example are provided in \cref{fig:betaeffectmultobs}, from where we notice comparable observations as in the single-obstacle case, with the controller driving each agent locally to the goal while ensuring that they avoid obstacles and inter-agent collisions en route. From the magnified view of the highlighted area in \cref{fig:betaeffectmultobs}, we also notice that $\hsighat$ admits a zero value for the LSE approximation scheme, which may cause the controller to revert to its nominal value, hence violating the safety constraints and effectively invalidating such CBF. $\hsighat$ is also less sensitive to the smoothing parameter for the polynomial approximation scheme than for the LSE technique. \hlrv{Finally, we notice from \cref{fig:trajenv}, that under the proposed controller, assured safety is guaranteed for all 100 experiment runs}.
\subsection{Effect of \texorpdfstring{$\beta$}{\textbeta} on Safety-Critical Control}\label{ssec:ubetaeffect}
Next, in \cref{fig:usafe}, we compare the nominal and CBF-filtered control inputs for $\beta \in \{0.1, 0.5, 1\}$, (for the multiple-obstacle setting only, due to page limits). From the plot, we observe that the polynomial smoothing technique tends to generate more conservative control actions that are adequate for the task and closely track the nominal CBF controller, modifying it only minimally, while that of the log-sum-exp function is more aggressive, exhibiting larger deviations from the nominal CBF-based controller, for unsafe control actions.
\begin{figure}[htb]
    \centering
    \hspace{-30pt}
    \def\mag{1.25}
    \def\spyw{1.5cm}
    \def\spyh{.7cm}
    \def\sz{1cm}
    \centering
    \begin{tikzpicture}[spy scope={magnification=\mag, size=\sz, width =\spyw, height=\spyh},every spy in node/.style={chamfered rectangle, fill=white, draw, dashed, red, very thick, cap=round}]
            \node {\includegraphics[width=.92\columnwidth]{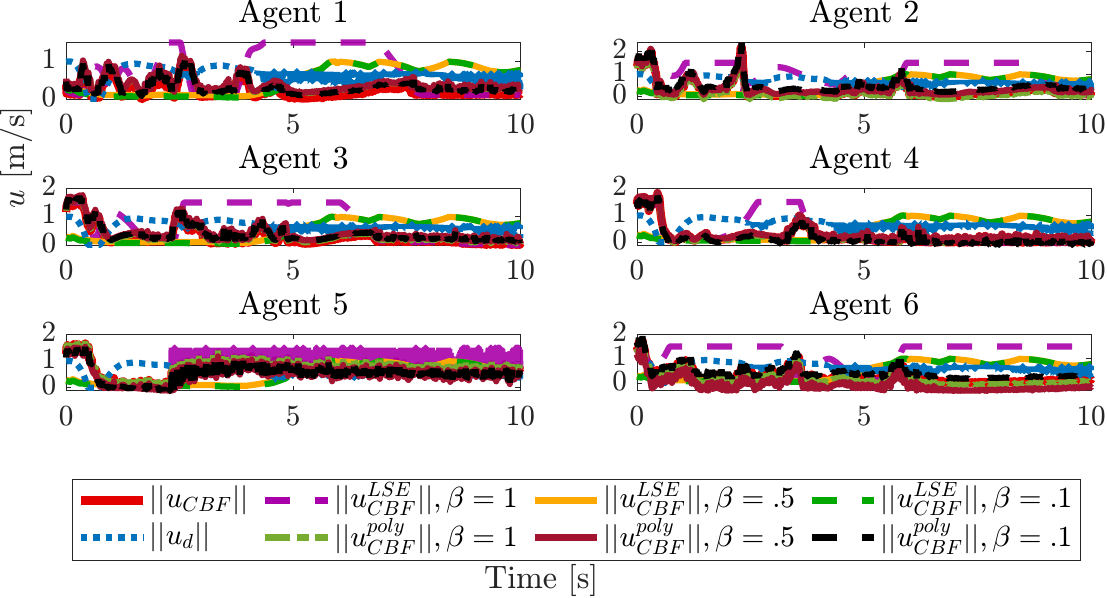}};
             \spy on (-.8,1.6) in node at (-1,1.65);
            \end{tikzpicture}
    \caption{\textbf{Synthesized control inputs}: Time-evolution of the norms of the nominal and CBF-filtered control inputs with $\beta=1, 0.5$, and $0.1$, and for the multiple-obstacle mission setting. \hlrv{A red inset highlighting the control inputs corresponding to the nominal and CBF-filtered cases is also depicted.} }
    \label{fig:usafe}
\end{figure}%
\subsection{Control Input Computation Time}\label{ssec:comptime}
Lastly, given that \cref{alg:ncbfps} relies on polynomial approximation and solving an online optimization problem, we computed average control input computation times for 100 simulation runs (see \cref{tab:timecomplx}) to roughly assess our method's performance. Although limited, the results show a direct relationship between the control horizon and average control synthesis time. Simulations were conducted using a workstation equipped with an AMD 7950x 16-core processor and 128 GB of RAM, with software programs optimized for parallel processing.
\begin{table}
\renewcommand{\arraystretch}{1.15} 
\setlength\arraycolsep{2pt}
    \centering
    \caption{\sc Computation Time (in Milliseconds)}
    \label{tab:timecomplx}
    \begin{tabular}{C{1.5cm}C{2cm}C{2cm}}
    \toprule
    $\tcont$  & \multicolumn{2}{c}{Computation Time [ms]}\\
    \cline{2-3}
    [time steps] & I & II\\
        \midrule
        \midrule
       10 & \RaggedRight$0.6 \pm 0.03$ &  \RaggedRight$ 1.05 \pm 0.2$ \\
       20 & \RaggedRight$0.647 \pm 0.035$ &  \RaggedRight$ 1.08 \pm 0.23$\\
       30 &\RaggedRight $0.66 \pm 0.037$ & \RaggedRight$1.11 \pm 0.239$\\
       \bottomrule
    \end{tabular}
\end{table}

\section{Conclusions}
\label{sec:conc}
This paper set out to illuminate an approach for safely coordinating autonomous collectives operating in environments with obstacles and under noisy control inputs. We showed that by polynomial smoothing of Boolean-composed CBFs and stochastic MPC, safe controls can be synthesized (locally) for collectives by considering the ensemble as one agent, computing controls, and propagating the resulting optimal values (component-wise) to individual agents, leading to both safety and performance guarantees. In addition, we provided evidence for the robustness of our proposed method to variations in approximation for small-enough values of the approximation parameter, with a fractional error upper bound. Our work, however, relies on the assumption of identical agent dynamics and a static environment, making controller synthesis for uncertain heterogeneous dynamical systems under varied tasking an interesting topic for future work. \hlrv{Other valid directions include an extension of the foregoing ideas to the non-Gaussian noise scenario or for more complex and dynamic mission environments. One can also consider gradient-free metaheuristic techniques that depart altogether from the QP-based framework presented here.}
\bibliographystyle{ieeetr}
\bibliography{root}%

\begin{thebibliography}{1}

\bibitem{glotfelterNonsmoothBarrierFunctions2017}
P.~Glotfelter, J.~Cortes, and M.~Egerstedt, ``Nonsmooth {Barrier} {Functions}
  {With} {Applications} to {Multi}-{Robot} {Systems},'' {\em IEEE Control
  Systems Letters}, vol.~1, pp.~310--315, Oct. 2017.

\bibitem{molnarComposingControlBarrier2023}
T.~G. Molnar and A.~D. Ames, ``Composing {Control} {Barrier} {Functions} for
  {Complex} {Safety} {Specifications},'' Sept. 2023.
\newblock arXiv:2309.06647 [cs, eess, math].

\bibitem{so2023almost}
O.~So, A.~Clark, and C.~Fan, ``Almost-sure safety guarantees of stochastic
  zero-control barrier functions do not hold,'' {\em arXiv preprint
  arXiv:2312.02430}, 2023.

\bibitem{ames_control_2019}
A.~D. Ames, S.~Coogan, M.~Egerstedt, G.~Notomista, K.~Sreenath, and P.~Tabuada,
  ``Control {Barrier} {Functions}: {Theory} and {Applications},'' in {\em 2019
  18th {European} {Control} {Conference} ({ECC})}, pp.~3420--3431, June 2019.

\bibitem{ames_control_2017}
A.~D. Ames, X.~Xu, J.~W. Grizzle, and P.~Tabuada, ``Control {Barrier}
  {Function} {Based} {Quadratic} {Programs} for {Safety} {Critical}
  {Systems},'' {\em IEEE Transactions on Automatic Control}, vol.~62,
  pp.~3861--3876, Aug. 2017.

\bibitem{ahmadi_safe_2019}
M.~Ahmadi, A.~Singletary, J.~W. Burdick, and A.~D. Ames, ``Safe {Policy}
  {Synthesis} in {Multi}-{Agent} {POMDPs} via {Discrete}-{Time} {Barrier}
  {Functions},'' in {\em 2019 {IEEE} 58th {Conference} on {Decision} and
  {Control} ({CDC})}, pp.~4797--4803, Dec. 2019.
\newblock ISSN: 2576-2370.

\bibitem{sahiner2018smoothing}
A.~Sah{\.i}ner, N.~Y{\.i}lmaz, and S.~A. Ibrahem, ``Smoothing approximations to
  non-smooth functions,'' {\em Journal of Multidisciplinary Modeling and
  Optimization}, vol.~1, no.~2, pp.~69--74, 2018.

\bibitem{mesbahStochasticModelPredictive2016}
A.~Mesbah, ``Stochastic model predictive control: An overview and perspectives
  for future research,'' {\em {IEEE} Control Systems Magazine}, vol.~36, no.~6,
  pp.~30--44, 2016.

\end{thebibliography}

\section*{Appendices}\label{sec:appx}
\subsection{Proof of Proposition 1}\label{prf:herrbound}
\begin{proof}
    Since $\hat{e} = 0$ everywhere on $\mbr$ where $\phi(\ell)$ and $\phihat$ coincide (see \cref{lem:phihat}), and $\ell \le \beta$ on $\mathcal{I}_{\beta} = [-\beta, \beta]$, using the definition of the $\liabnorm$ norm and from \cref{eq:hsigalt,eq:hsigaltfin}, we can write%
     \begingroup
        \small
        \begin{equation*}
        \begin{split}%
        &\E{||\hat{e}||_{\liabnorm}} = \E{\int_{{-\beta}}^{\beta}{|\hat{e}|d\ell}}\\
        &=\E{ \int_{-\beta}^{\beta}{{\big|-\frac{1}{2}(\ell\hat{\phi} - \ell^{\prime}) + \frac{1}{2}(\ell\phi - \ell^{\prime})\big|}d\ell}}\\
        &= \E{ \int_{-\beta}^{\beta}{{\big|\frac{1}{2}\ell(\phi - \hat{\phi})\big|}d\ell}}\\
        & \le \frac{\beta}{2}\E{\int_{-\beta}^{\beta}{| \phi -  \hat{\phi}|} d\ell } \ \text{\small {(by the monotonicity of $\E{\cdot}$)}}\\
        &= \frac{\beta}{2}\E{|| \phi - \hat{\phi}||_{\liabnorm}}= \frac{\beta}{2}\cdot{\frac{15\beta}{4}} = \frac{15\beta^2}{8},
           \end{split}
        \end{equation*}%
    \endgroup
    \normalsize
     by \cref{lem:phihat}. \hlrv{Thus, we arrive at the intended result given in Proposition 1 (i.e., $\E{||\hat{e}||_{\liabnorm}} \le \frac{15\beta^2}{8}$), which completes the proof.} 
\end{proof}%
\hlrv{
\subsection{Proof of Proposition 2}\label{prf:prop2cbfsafe}
\begin{proof}
    Let $\mrmdrift$ and $\mrmdiff$ denote the ensemble analogs of $\mu_i$ and $\sigma_i$, defined respectively as
    \begin{subequations}
        \begin{align}
        \mrmdrift &:= L_F\Hat{h}_\Sigma(\mrmx) + L_G\Hat{h}_\Sigma(\mrmx)\mathrm{u} \nonumber \\
        & \quad + \frac{1}{2}\operatorname{tr}\Big(\big((K_w\otimes I_N)^{\top}\frac{\partial^2\hsighat(\mrmx)}{\partial \mrmx^2}(K_w\otimes I_N)\big)\Big)\\
        \mrmdiff &:= \frac{\partial \hsighat(\mrmx)}{\partial \mrmx}(K_w(\mrmx)\otimes I_N).
        \end{align}
    \end{subequations}
    According to \cref{thm:soalmostsureconv}, we need to show that $\hsighat$ satisfies 
    \begin{equation}
        \label{eq:relsigmu}
        \mrmdrift-\frac{{\mrmdiff}^2}{\hsighat(\mrmx)} \geq-\hsighat^2(\mrmx) \alpha_\star(\hsighat).
    \end{equation}
    Our proof will proceed in two steps. First, we will show that $\hsighat$ is a valid zeroing CBF (in Part I of the proof). To finish the proof, in Part II, we will then construct a \textit{reciprocal} CBF (RCBF) from $\hsighat$ and prove that such an RCBF satisfies valid conditions for almost-sure safety as set forth in Theorem 10 in \cite{so2023almost}. Part I's proof follows.
    
    From \cref{eq:phihat,eq:hsigaltfin}, by the definitions of $\ell$ and $\ell^\prime$, and since $\inf_{\ell}{M_2(\ell, \beta)} = -1$ and $\sup_{\ell}{M_2(\ell, \beta)} = 1$ on $-\beta \le \ell \le \beta$, it can be shown that $\hsighat > 0$ for all $h_{*}, * \in \mathcal{I}$ such that $\forall \ \mrmx \in \partial\mf{C}_\Sigma, h_*(\mrmx)$ is positive\footnote{This statement is valid for all $h_*(\mrmx)$ cases, i.e., for $\hdiff \ge 0$ and for $\hdiff \le 0$. In all cases, $\ell^\prime > 0$, since $h_* > 0,$ for all $* \in \mathcal{I}$ by \cref{def:cbf}.}. Note that, by the sum and difference law of limits and the continuous differentiability of its constituent CBFs, $\hsighat$ is smooth and twice continuously-differentiable, i.e., in $C^2$. Thus, by Lemma 2 in \cite{ames_control_2017}, if $\hsighat(\mrmx) > 0 \ \forall \ \mrmx \in \partial\mf{C}_\Sigma$, then for each $k\in \mathbb{Z}_+$, there exists a constant $\gamma > 0$ s.t. 
    \begin{equation}
        \label{eq:lem2ames2017}
        \frac{d \hsighat(\mrmx)}{d \mrmx} \ge -\gamma \hsighat^k(\mrmx), \ \forall \ \mrmx \in \operatorname{Int}(\mf{C}).
    \end{equation}
    Setting $\alpha(\gamma, \hsighat) = \gamma \hsighat^k(\mrmx)$, we conclude that $\hsighat$ satisfies the ensemble analog of \cref{eq:cbfppty}, thus completing Part I's proof. 

    For Part II, motivated by Corollary 11 in \cite{so2023almost}, we define the RCBF, $\Bsighat := {1}/{\hsighat}$, and write its corresponding drift term, $\Bmrmdrift$, in terms of $\mrmdrift$ and $\mrmdiff$ as
    \begin{equation}
        \label{eq:rcbfinter}
        \begin{aligned}
        \Bmrmdrift =-\hsighat^{-2} \mrmdrift+\hsighat^{-3} \mrmdiff^2.
        \end{aligned}
    \end{equation}
    By the definition of an RCBF (see \cite{ames_control_2017, so2023almost}), from \cref{eq:rcbfinter}, we can write
    \begin{equation*}
    \begin{aligned}
    -\hsighat(\mrmx)^{-1}\Bmrmdrift &\leq \alpha_\star(\hsighat)\\
    \implies-\hsighat(\mrmx)^{-1}\left(-\hsighat^{-2} \mrmdrift+\hsighat^{-3} \mrmdiff^2\right)&\leq \alpha_\star(\hsighat)\\
    \label{seq:b}
    \implies\mrmdrift-\frac{{\mrmdiff}^2}{\hsighat(\mrmx)} &\geq-\hsighat^2(\mrmx) \alpha_\star(\hsighat),
    \end{aligned}
    \end{equation*}
    which is the same expression in \cref{eq:relsigmu}, where $\alpha_\star$ is a $\mck$-class function. We have now proved almost-sure safety for $\hsighat$. Invoking Proposition 1 in \cite{ames_control_2017}, since the previous argument establishes the validity of $\hsighat$ as a zeroing CBF for all time $t\ge 0$, the forward invariance of $\mf{C}_\Sigma$ follows immediately, thus completing the proof.
\end{proof}
}
\end{document}